\documentclass[onecolumn,floatfix,ams,nofootinbib]{revtex4}
\usepackage[dvips]{epsfig}
\usepackage[english]{babel}
\usepackage{bbm}
\usepackage{verbatim}
\usepackage{array}
\usepackage{bm} 
\usepackage{amsmath}
\usepackage{amsthm}
\usepackage{amsfonts}
\usepackage{amssymb}
\newtheorem{definition}{Definition}
\usepackage{hyperref}
\usepackage{cases}
\usepackage[dvipsnames]{xcolor}
\hypersetup{colorlinks=true,linkbordercolor=Blue,linkcolor=Blue, citecolor=Blue}
\usepackage{indentfirst} %primeiro paragrafo com espaço
\usepackage[titletoc]{appendix}
\usepackage{subeqnarray}
\usepackage{setspace}
\usepackage{indentfirst} %primeiro paragrafo com espaço
\usepackage{gensymb}%simbolo º
\usepackage{xcolor}%Pacote de mudança de cores nos caracteres 
\usepackage{calrsfs}%Fonte do operador Dilatação

\usepackage{wasysym}
\usepackage[all,cmtip]{xy}

    \usepackage{amsmath,amsfonts,amssymb,amsthm,mathrsfs,bbm,braket}

    \newtheorem{prop}{Proposition}[section]

    \numberwithin{equation}{section}

\usepackage{bbold}

\begin{document}

\title{Deformations in spinor bundles: Lorentz violation and further physical implications} 

\author{J. M. Hoff da Silva\footnote{Corresponding author.}} 
\email{julio.hoff@unesp.br}
\affiliation{Departamento de F\'isica, Universidade
Estadual Paulista, UNESP, Av. Dr. Ariberto Pereira da Cunha, 333, Guaratinguet\'a, SP,
Brazil.}

\author{R. T. Cavalcanti} 
\email{rogerio.cavalcanti@ime.uerj.br}
\affiliation{Departamento de Matem\'atica Aplicada, Universidade do Estado do Rio de Janeiro, UERJ, Rua São Francisco Xavier, 524, Maracanã, Rio de Janeiro, RJ , Brazil.}
\affiliation{Departamento de F\'isica, Universidade
Estadual Paulista, UNESP, Av. Dr. Ariberto Pereira da Cunha, 333, Guaratinguet\'a, SP,
Brazil.}

\author{G. M. Caires da Rocha} 
\email{gabriel.marcondes@unesp.br}
\affiliation{Departamento de F\'isica, Universidade
	Estadual Paulista, UNESP, Av. Dr. Ariberto Pereira da Cunha, 333, Guaratinguet\'a, SP,
	Brazil. }
%\date{\today}

%\pacs{}
%\keywords{}

\begin{abstract}
This paper delves into the deformation of spinor structures within nontrivial topologies and their physical implications. The deformation is modeled by introducing real functions that modify the standard spinor dynamics, leading to distinct physical regions characterized by varying degrees of Lorentz symmetry violation. It allows us to investigate the effects in the dynamical equation and a geometrized nonlinear sigma model. The findings suggest significant implications for the spinor fields in regions with nontrivial topologies, providing a robust mathematical approach to studying exotic spinor behavior.
\end{abstract}		

\maketitle

\section{Introduction}

It is well known that spinors play a central role in modern theoretical physics and geometry \cite{1,PS}. Spinor fields, which are sections of spinor bundles \cite{var2}, are essential in describing fermionic fields, for example. Among the different approaches to describing such objects, it should be stressed the one proposed by \'Elie Cartan \cite{4}. Cartan's approach views spinors as projections from the Riemann sphere, obtained by intersecting the light cone at a specific time onto the complex plane. This allows a correspondence between spacetime points $(t, x, y, z)$ and spinors entries $(\zeta,\eta)$,  such that $(t, x, y, z) = (t(\zeta,\eta), x(\zeta,\eta), y(\zeta,\eta), z(\zeta,\eta))$.  However, some intriguing phenomena arise when the bundles are defined over manifolds with nontrivial topology. These include the possibility of having nonequivalent spinor structures \cite{pet,6}, which can lead to the emergence of exotic spinors with distinct physical properties from their standard counterparts \cite{var3,8,var4,np}.

Our motivation here stems from the need to understand how deformations in spinor bundles, induced by the underlying topology of the manifold, affect the dynamics of the corresponding spinor fields. These deformations are characterized by introducing a smooth, real function $\varphi(x)$, which modulates the transition functions between different coordinate patchs on the manifold. The deformation function $\varphi(x)$, which varies smoothly across the manifold, plays a crucial role in distinguishing different physical regions and determining the extent to which Lorentz and Poincar\`e symmetries are possibly violated. It is of particular interest in both theoretical and experimental physics, as they can have potentially observable consequences in, but not restricted to, high-energy physics \cite{8,var4,np,11}.

The study begins with a review of the fundamental mathematical tools required to understand spinor structures on manifolds, including the algebraic and topological foundations of spinor bundles, mainly focusing on the conditions under which multiple, nonequivalent spinor structures can exist. This discussion is grounded in the framework of Clifford algebras and the Spin group, which are pivotal in defining spinor representations in different dimensions and signatures. By examining the interplay between the algebraic properties of spinors and the topological characteristics of the manifold, we lay the groundwork for analyzing the deformations introduced by the $\varphi(x)$ function.

A key aspect of our investigation is the partitioning of the base manifold $\mathcal{M}$ into three distinct regions: $\Sigma_1$, a multiply connected region with nontrivial topology; $\Sigma_2$ is a region subtly influenced by the deformation function; and $\Sigma_3$, a region with constant deformation function $\varphi(x)$. This partitioning allows for a precise analysis of how the deformation affects the spinor fields in different manifold parts. In particular, the behavior of the deformation function $\varphi(x)$ across these regions provides insights into how the topological properties of $\mathcal{M}$ influence the physical properties of the spinor fields.

The formalism adopted in this paper is further applied to study the physical consequences of these deformations. We explore how the deformation function $\varphi(x)$ affects the spinor current's conservation analysis, the spinor fields' dynamics, and the potential violations of Lorentz symmetries. The results of this analysis are not only of theoretical interest but also have potential implications for understanding the behavior of fermions in curved spacetimes \cite{12,13,penro}, as well as in scenarios involving symmetry-breaking mechanisms in high-energy physics \cite{15,16}.

This work is organized as follows: In Sections II and III, we shall delve into the mathematical fundamentals of the proposed deformation, which provides a foundation for understanding the interplay between topology, algebra, and geometry in the context of spinor bundles. Section IV explores some relevant implications of the inequivalence of spinor structures and the resulting physical phenomena. This section is divided into two parts: the first explores kinematic and dynamic effects through quasinormal-like behavior and the Gordon decomposition. In contrast, the second investigates the geometrical implications of nontrivial topology by analyzing a nonlinear sigma model. In Section V, we conclude. Finally, in the Appendix, we discuss relevant aspects of the Hausdorff measure supporting the considerations made in Section IV. We stress that in Section III and the Appendix, Greek letters denote exclusively open sets, while in Section II, their usage may be straightforwardly inferred from the context.

\section{Basic tools and deformation}

In this section, we shall start giving a very brief account of spinor structures within the nontrivial topology context (for a more comprehensive presentation of the preliminary tools, see \cite{var,var2,var3,var4}), and after we present the basic steps of the proposed deformation.  
 
Let $\mathcal{M}$ be a base manifold and $\bigcup_{\alpha\in \mathbb{N}^*} U_\alpha$ a covering of $\mathcal{M}$ by open sets. If $\mathcal{M}$ has a nontrivial topology, more than one nonequivalent spinor structure may arise. For instance, for all $x\in U_\alpha\cap U_\beta\subset \mathcal{M}$ there are at least two mappings, say $h_{\alpha\beta}(x)$ and $\tilde{h}_{\alpha\beta}(x)$, from the intersection above to\footnote{Please see Ref. \cite{KA} for the adequacy of classes for pseudo-Riemannian manifolds of any dimension.} $Spin(1,3)$. The mappings are related by $\tilde{h}_{\alpha\beta}(x)=h_{\alpha\beta}(x)c_{\alpha\beta}(x)$, where $c_{\alpha\beta}(x)\in \mathbb{Z}_2\hookrightarrow Spin(1,3)$ are the transition functions. In this way, there are also at least two orthonormal coframe fiber bundles of which nonequivalent spinor fields are sections, i.e., $\sec P_{Spin(1,3)}\times_\sigma \mathbb{C}^4\ni \psi$ and $\sec \tilde{P}_{Spin(1,3)}\times_\sigma \mathbb{C}^4\ni \tilde{\psi}$, where $\sigma=\{(1/2,0)\oplus(0,1/2),(1/2,0),(0,1/2)\}$. The spinor $\tilde{\psi}$ is called exotic. The $Spin(1,3)$ group may be understood as the subgroup of invertible elements in $Cl(1,3)$. Hence, consider $\bar{\rho}$ as the linear, faithful and invertible map $\bar{\rho}:Cl(1,3)\rightarrow M(4,\mathbb{C})$ and since\footnote{Where $\mathbb{H}$ stands for the quaternions and $Cl^+(1,3)$ denotes the even sub Clifford algebra.} $\mathbb{C}\subset \mathbb{H}\simeq Cl^+(1,3)\subset Cl(1,3)$, admit the existence of $U(1)$ functions \cite{pet} $\xi_\alpha:U_\alpha\rightarrow \mathbb{C}$ such that 
\begin{eqnarray}
\quad \qquad \qquad \qquad \qquad 
\xi_\alpha(x)=\bar{\rho}(c_{\alpha\beta}(x))\xi_\beta(x),\qquad \forall \; x\!\in U_\alpha\cap U_\beta. \label{s0}
\end{eqnarray} Given that $c_{\alpha\beta}=\pm e$ and $\bar{\rho}$ is faithful, then $\bar{\rho}(c_{\alpha\beta}(x))=\pm 1$. Besides, locally, spinors are also connected (in the same fiber bundle) by $\bar{\rho}(h_{\alpha\beta})$ \cite{var3}
\begin{equation}
\psi_\alpha=\bar{\rho}(h_{\alpha\beta})\psi_\beta
\end{equation} and similarly for exotic spinors, that is $\tilde{\psi}_\alpha=\bar{\rho}(\tilde{h}_{\alpha\beta})\tilde{\psi}_\beta$, from which $\tilde{\psi}_\alpha=\bar{\rho}(h_{\alpha\beta}c_{\alpha\beta})\tilde{\psi}_\beta$. Since $\bar{\rho}$ is a linear mapping with a range in a group, one has $\bar{\rho}(ab)=\bar{\rho}(a)\bar{\rho}(b)$ and, therefore, 
\begin{equation}
\tilde{\psi}_{\alpha}=\bar{\rho}(h_{\alpha\beta})\bar{\rho}(c_{\alpha\beta})\tilde{\psi}_\beta.\label{s1}
\end{equation} Now, from Eq. (\ref{s0}), we get $\bar{\rho}(\xi_\beta)=\bar{\rho}(c_{\alpha\beta})\bar{\rho}(\xi_\alpha)$ and, by means of Eq. (\ref{s1}), we are left with
\begin{equation}
\bar{\rho}(\xi_\alpha)\tilde{\psi}_\alpha=\bar{\rho}(h_{\alpha\beta})\bar{\rho}(\xi_\beta)\tilde{\psi}_\beta.
\end{equation} Comparing this last expression with $\psi_\alpha=\bar{\rho}(h_{\alpha\beta})\psi_\beta$ we have $\psi_\alpha=\bar{\rho}(\xi_\alpha)\tilde{\psi}_\alpha$ and suppressing the (same) open set index we have the mapping $\psi=\bar{\rho}(\xi)\tilde{\psi}$, that is:
\begin{eqnarray}\label{a0}
\bar{\rho}(\xi)&:&\left.\tilde{P}\rightarrow P \right.\nonumber\\ &&\left.\tilde{\psi}\mapsto \bar{\rho}\tilde{\psi}=\psi.\right.
\end{eqnarray} 

Now, it is straightforward to see that 
\begin{eqnarray}\label{a1}
\gamma^\mu\partial_\mu\psi=\gamma^\mu\partial_\mu\bar{\rho}\tilde{\psi}+\gamma^\mu\bar{\rho}\partial_\mu\tilde{\psi}.
\end{eqnarray} One can tentatively define $B:\tilde{P}\rightarrow\tilde{P}$ acting as a (kind of) compensating form to guarantee the right transformation of the spinor field derivative. The point is that the spinor derivative is also a spinor and, as such, must retain the transformation law (\ref{a0}). Hence, by setting $\nabla_\mu=\partial_\mu+B_\mu$ we have 
\begin{eqnarray}\label{a2}   
\bar{\rho}\gamma^\mu\nabla_\mu\tilde{\psi}=\bar{\rho}\gamma^\mu\partial_\mu\tilde{\psi}+\bar{\rho}\gamma^\mu B_\mu\tilde{\psi}.
\end{eqnarray} Subtracting Eq. (\ref{a1}) from (\ref{a2}) yields 
\begin{eqnarray}\label{a3}
\bar{\rho}\gamma^\mu\nabla_\mu\tilde{\psi}-\gamma^\mu\partial_\mu\psi=[\bar{\rho},\gamma^\mu]\partial_\mu\tilde{\psi}+(\bar{\rho}\gamma^\mu B_\mu-\gamma^\mu\partial_\mu\bar{\rho})\tilde{\psi}.
\end{eqnarray} The imposition $\bar{\rho}\gamma^\mu\nabla_\mu\tilde{\psi}=\gamma^\mu\partial_\mu\psi$ (just as $\bar{\rho}\tilde{\psi}=\psi$) can now be appreciated from the vanishing of Eq. (\ref{a3}) right-hand side. It is worthwhile to mention that, by now, the scalar and vector bilinear covariants of the spinor and its exotic counterpart are related by 
\begin{eqnarray}
\bar{\psi}\psi=\bar{\tilde{\psi}}(\gamma^0\bar{\rho}^\dagger\gamma^0\bar{\rho})\tilde{\psi},\nonumber\\
j^\mu=\bar{\tilde{\psi}}(\gamma^0\bar{\rho}^\dagger\gamma^0\gamma^\mu\bar{\rho})\tilde{\psi}.\label{a4}
\end{eqnarray} 

We shall cast sufficient conditions to the vanishing of (\ref{a3}). Firstly, 
\begin{eqnarray}
\hspace{1.5cm} 
[\bar{\rho},\gamma^\mu]=0, \hspace{.5cm} \forall \mu,\label{a5}
\end{eqnarray} from which $[\gamma^\mu,\bar{\rho}^{-1}]=0=[\gamma^\mu,\bar{\rho}^\dagger]$ are immediate, as well $\bar{\psi}\psi=\bar{\tilde{\psi}}|\bar{\rho}|^2\tilde{\psi}$ and $j^\mu=\bar{\tilde{\psi}}|\bar{\rho}|^2\gamma^\mu\tilde{\psi}$ for a mostly negative metric signature. Then, taking into account (\ref{a5}), it is possible to write the $B$ form coefficient as $B_\mu=\bar{\rho}^{-1}\partial_\mu\bar{\rho}$ and hence
\begin{equation}
\nabla_\mu=\partial_\mu+(\bar{\rho})^{-1}\partial_\mu\bar{\rho}. \label{cor}
\end{equation} At this point, taking a step back for a last consideration on the open set index is convenient. First, as we shall discuss in a moment, $\bar{\rho}$ will suffer a deformation by a globally defined function $\varphi(x)$ so that $\bar{\rho}=\varphi(x)\rho_\alpha$. Notice that $\rho_\alpha^2=\rho_\beta^2$ in $U_\alpha\cap U_\beta$ and hence it is worth working with the squares of $\rho_\alpha$ since they provide an everywhere continuous function\footnote{We are denoting $\rho(\xi_\alpha)$ by $\rho_\alpha$.}. It is then said that $c_{\alpha\beta}$ are generated by the square roots of $\rho_\alpha^2$. Calling $\rho_\alpha^2\equiv\rho$ we have, for the $B_\mu$ form in Eq. (\ref{cor}), 
\begin{equation}
\rho_\alpha^{-1}\partial_\mu\rho_\alpha=\rho^{-1}\rho_\alpha(\partial_\mu\rho\,\rho_\alpha^{-1}+\rho\partial_\mu\rho_\alpha^{-1}),
\end{equation} yielding $\rho_\alpha^{-1}\partial_\mu\rho_\alpha=\frac{1}{2}\rho^{-1}\partial_\mu\rho$. Therefore, we shall replace Eq. (\ref{cor}) with the following general expression valid everywhere
\begin{equation}
\nabla_\mu=\partial_\mu+\varphi^{-1}\partial_\mu\varphi+\frac{1}{2}\rho^{-1}\partial_\mu\rho. \label{cor2}
\end{equation}

Usually, the solution to (\ref{a5}) involves expressing $\rho$ in terms of convenient $U(1)$ functions, $\xi(x)=\exp^{i\theta(x)}$, along with the invocation of Schur's lemma (thinking of gamma matrices as a base to Clifford algebras) to set the matrix status of $\rho$ as given by the identity in $M(4,\mathbb{C})$. This is indeed a typical Spin transformation, leading to  $\bar{\psi}\psi=\bar{\tilde{\psi}}\tilde{\psi}$ and the current invariance \cite{lich}. That is precisely the point of departure in our approach, i.e., as already mentioned, we shall deform the $\rho$-mapping by inserting a smooth real function $\varphi(x):U\subset\mathcal{M}\rightarrow (0,1]$ whose behavior shall set distinct physical regions (using appropriate conditions). The central point in this proposed deformation is that with $\varphi(x)$ the mapping $\rho$ is not unimodular. In light of our exposition, the price to pay is a violation of Lorentz symmetries everywhere $\varphi(x)\neq 1$. In the course of developments, it is only natural to study the violation of Poincar\`e symmetries. This last part of the deformation is performed by a complementary behavior of the $\theta(x)$ function, helping to localize (or limit) the effects of nontrivial topology; this is a subtle aspect whose implementation borders on mathematical malpractice, but we shall discuss that in a moment, shortly after deepening studies on current conservation. 

Being $\bar{\rho}(x)=\varphi(x)\exp^{i\theta(x)}\mathbb{1}_{M(4,\mathbb{C})}$, we have 
\begin{equation}
	\partial_\mu j^\mu=\partial_\mu \varphi^2(x)\bar{\tilde{\psi}}\gamma^\mu\tilde{\psi}+\varphi^2(x)\partial_\mu\bar{\tilde{\psi}}\gamma^\mu\tilde{\psi}+\varphi^2(x)\bar{\tilde{\psi}}\gamma^\mu\partial_\mu\tilde{\psi}. \label{a6}
\end{equation} By the same token, the dynamical equation obeyed by the exotic spinor reads $i\gamma^\mu(\partial_\mu+\varphi^{-1}\partial_\mu\varphi+\frac{1}{2}\rho^{-1}\partial_\mu\rho)\tilde{\psi}-m\tilde{\psi}=0$, and therefore it is possible to find 
\begin{equation}
\gamma^\mu\partial_\mu\tilde{\psi}=-\frac{i}{2}\gamma^\mu\partial_\mu\theta(x)\tilde{\psi}-\varphi^{-1}(x)\gamma^\mu\partial_\mu \varphi(x)\tilde{\psi}-im\tilde{\psi}\label{a8} 
\end{equation} and similarly 
\begin{equation}
\partial_\mu\bar{\tilde{\psi}}\gamma^\mu=\frac{i}{2}\partial_\mu\theta(x)\bar{\tilde{\psi}}\gamma^\mu-\varphi^{-1}(x)\bar{\tilde{\psi}}\gamma^\mu\partial_\mu \varphi(x)+im\bar{\tilde{\psi}}.\label{a9} 
\end{equation} Reinserting Eqs. (\ref{a8}) and (\ref{a9}) back into (\ref{a6}) it is readily verified that $\partial_\mu j^\mu=0$, as expected. 

Here is a good point to call attention to the fact that, despite the form of $\rho^{-1}\partial_\mu \rho$ correction in (\ref{cor}), there is no gauge transformation absorbing it. The reason is the following: in the dynamical equation for the exotic spinor (\ref{a8}), we have two additional terms, one coming from the topological correction and another from the mapping deformation. The term $\gamma^\mu\partial_\mu\theta(x)$ cannot be eliminated by a gauge transformation since there is no gauge function, say $\Lambda$, with the topological $\theta(x)$ function properties (according to the previous discussion) \cite{pet,CI}, that is, $\Lambda\sim \ln\rho_\alpha$. On the other hand, for the second term $i\varphi^{-1}(x)\gamma^\mu\partial_\mu\varphi(x)$, it is clear that a tentative gauge function would have to be complex to absorb the deformation properly. Therefore, both contributions survive. 

We shall finalize this section by setting the picture of topological non-triviality we are interested in this paper: consider a partition of the base manifold $\mathcal{M}$ into three regions, $\Sigma_1$ which is multiply connected (to fix ideas one could consider $\Sigma_1 \supseteq (\mathbb{R}^2\times S^1)\times \mathbb{R}$), $\Sigma_2\simeq \mathbb{R}^{1,3}\backslash (\Sigma_1\cup \Sigma_3)\equiv \,^\varphi\mathbb{R}^{1,3}$ a region with trivial topology but in which there is still some effect coming from the deformation procedure, and $\Sigma_3\simeq \mathbb{R}^{1,3}$ a usual part $\mathcal{M}$. In order to localize, so to speak, the nontrivial topology region and separate the deformation effects, we found suitable the following boundary conditions: 

\begin{enumerate}
	\item $\theta(x)$ varying and $\varphi(x)=1$, $\forall\, x\in\, \Sigma_1$. 
	
In this nontrivial topology region the local square roots of $\rho(x)$ may change sign abruptly by crossing surfaces $S\subset U_\alpha\cap U_\beta$, that is $\xi_\alpha/\xi_\beta=\exp^{i(\theta_\alpha-\theta_\beta)}=\pm 1$ and the fermionic current is constant.    

	\item $\theta(x)$ constant and $\varphi(x):U\subset\mathcal{M}\rightarrow (0,1)$ varying in general,  $\forall\, x\in\, ^\varphi\mathbb{R}^{1,3}$. 
	
In this region, we have a deformed mapping $\tilde{P}_{Spin(1,3)}\times_\sigma\mathbb{C}^4\underbrace{\longrightarrow}_{\rho\circ\varphi\circ x} P_{Spin(1,3)}\times_\sigma\mathbb{C}^4$	but as $\theta$ is constant we have a lack of contribution of this term to the dynamical operator. From the physical point of view, however, it seems unreasonable that a given area with nontrivial topology shall affect the fermionic behavior everywhere. Then, we shall impose $\theta$ becoming gradually constant as it reaches the boundary of $\Sigma_1$. We also shall take the deformation $\varphi(x)$ as a slowly varying function. Physically, it means a handle Lorentz symmetry violation and is implemented by requiring that $\partial_\mu\ln{\varphi(x)}$ is a constant vector of low magnitude. The impact of this region awkwardness study is the main focus of this paper. We shall furnish a mathematical criteria of slowly varying $\varphi(x)$ in the Appendix.  

	\item $\theta(x)$ and $\varphi(x)$ are constants $\forall\, x\in\, \Sigma_3$.    
	
Here, the eventual residue of $\varphi$ constant in the bilinear covariants shall also play a hole in the analysis (see Section IV). One could argue that a simple spinor redefinition may absorb it. However, we shall frame the discussion based upon spinors whose normalization coincides if, and only if, the deformation is restored to $1$. 
\end{enumerate}	         

After properly characterizing the deformation mathematically in the next section, we shall further discuss in Section IV the current conservation given the separation between spacetime regions proposed above. 

\section{Mathematical Characterization}

We shall introduce some standard notions to proceed with a formal analysis in a general fashion after what we contextualize to the case at hand. This presentation (more general and robust) is more likely for our purposes. We parallel part of Ref. \cite{var} conceptualization. 

Let $\Gamma$ and $H$ be two topological groups and $f:\Gamma\rightarrow H$ a continuous homomorphism with $\ker{(f)}=\mathbb{Z}_2\subset Z(\Gamma)$, where $Z(G)$ denotes de center\footnote{Of course, for $\Gamma=Spin(1,3)$ and $H=SO(1,3)$ we have a direct contact with the previous section.} of $H$. Consider a principal bundle $P=(\mathcal{M},\pi,H)$. A $\Gamma$-structure in $P$ is a principal bundle (with $\Gamma$ as structure group), say $^\Gamma\!P$, along with an equivariant mapping 
\begin{eqnarray}
\eta&:&\left. \!\,^\Gamma \!P\rightarrow P\right.\nonumber\\ &&\left. 
\!\!\! 
l\cdot\gamma \mapsto\eta(l\cdot\gamma)=\eta(l)\cdot f(\gamma),\right.\label{n1}
\end{eqnarray} for $l\in \,^\Gamma \!P$ and $\gamma\in\Gamma$. A principal bundle admits $\Gamma$-structure if, and only if, there are continuous maps\footnote{Not to be confused with the other section's Dirac matrices.} $\gamma_{\alpha\beta}:U_\alpha\cap U_\beta\subset\mathcal{M}\rightarrow \Gamma$ such that $\gamma_{\alpha\beta}$ obey the consistency condition $\gamma_{\alpha\beta}(x)\gamma_{\beta\rho}(x)=\gamma_{\alpha\rho}(x)$ for $x\in U_\alpha\cap U_\beta \cap U_\rho$ and the composition $f\cdot\gamma_{\alpha\beta}$ leads to the transition functions in $P$ \cite{st}. Therefore, in general (whenever possible) transition functions $U_\alpha\cap U_\beta\rightarrow H$ are lifted to $\gamma_{\alpha\beta}:U_\alpha\cap U_\beta\rightarrow \Gamma$ and $\gamma_{\alpha\beta}$ are called the lifting of the structure group. 

As a final ingredient, consider $P$ a principal bundle with structure group $H$ and let be $\lambda:H\rightarrow H'$ a homomorphism. Let $\bigcup_{\alpha\in \mathbb{N}^*} U_\alpha$ be a suitable covering of the base manifold in $P$ with a complete system of local sections $s_\alpha$ and transition functions $t_{\alpha\beta}$. By defining $t'_{\alpha\beta}=\lambda \cdot t_{\alpha\beta}$, one has $t'_{\alpha\beta}(x)t'_{\beta\rho}(x)t'_{\rho\alpha}(x)=\lambda e_{H}=e_{H'}$, so that there is a corresponding principal bundle $P'$ with $t'_{\alpha\beta}$ as transition functions. The bundle $P'$ is said to be a $\lambda$-extension of $P$.  

Now we can construct criteria for equivalence of spinor structures and adapt it to our context. Consider $f$ as before and admit that $P$ has two $\Gamma$-structures, $\tilde{P}$ and $\tilde{P}'$. Let be $\gamma_{\alpha\beta}:U_\alpha\cap U_\beta\rightarrow \Gamma$ and $\gamma'_{\alpha\beta}:U_\alpha\cap U_\beta\rightarrow \Gamma'$ transition functions and set maps 
\begin{eqnarray}
\delta_{\alpha\beta}&:&\left.U_\alpha\cap U_\beta\rightarrow \Gamma\right.\nonumber\\ &&\left. x \mapsto \gamma_{\alpha\beta}(x)\gamma'\,^{-1}_{\alpha\beta}(x).\right.\label{u1}
\end{eqnarray} As a last step, define the homomorphism (the analog of $\lambda$ for $\Gamma$ and $\Gamma'$) $\bar{\lambda}:\Gamma'\rightarrow\Gamma$ by $\gamma_{\alpha\beta}(x)=\bar{\lambda}\gamma'_{\alpha\beta}(x)$. 
\begin{prop} 
	$\delta_{\alpha\beta}\in \mathbb{Z}_2$. 
\end{prop}
\begin{proof}
From the definition of $\delta_{\alpha\beta}(x)$, $f(\delta_{\alpha\beta}(x))=f(\gamma_{\alpha\beta}(x)\gamma'\,^{-1}_{\alpha\beta}(x))$, or $f(\delta_{\alpha\beta}(x))=f(\bar{\lambda}e_{\Gamma'})=f(e_{\Gamma})=e_{H}$ and therefore $\delta_{\alpha\beta}(x)\in \ker{(f)}$. 
\end{proof}
In this way, the $\delta_{\alpha\beta}$ mappings define a 1-cochain, say $\delta(\alpha,\beta)(x)$, for $x\in U_\alpha\cap U_\beta$. 

\begin{prop}
 $\delta$ is a cocycle. 
\end{prop}
\begin{proof}
 Let $\partial$ denote the coboundary operator, such that for $x\in U_\alpha\cap U_\beta\cap U_\kappa$ we have
\begin{equation} 
\partial \delta(\alpha,\beta,\kappa)=\delta_{\beta\kappa}\delta_{\alpha\kappa}^{-1}\delta_{\alpha\beta}.
\end{equation} Taking into account Eq. (\ref{u1}) and the definition of $\bar{\lambda}$, it is readily verified that 
\begin{equation}
\partial \delta(\alpha,\beta,\kappa)=(\bar{\lambda}\gamma'_{\beta\kappa}\gamma_{\beta\kappa}'^{-1})\gamma_{\alpha\kappa}'\gamma_{\alpha\kappa}^{-1}(\bar{\lambda}\gamma'_{\alpha\beta}\gamma_{\alpha\beta}^{-1}). 
\end{equation} The terms in parenthesis in the equation above are trivially mapped into $e_{\Gamma}$, whilst $\gamma_{\alpha\kappa}'\gamma_{\alpha\kappa}^{-1}=(\gamma_{\alpha\kappa}\gamma_{\alpha\kappa}'^{-1})^{-1}=(\bar{\lambda}\gamma_{\alpha\kappa}'\gamma_{\alpha\kappa}'^{-1})^{-1}=e_{\Gamma}$, that is $\partial \delta(\alpha,\beta,\kappa)=e_{\Gamma}$. 
\end{proof}
 
This last proposition shows that $\delta$ belongs to the first \v{C}ech cohomology group. The element $\delta(\tilde{P},\tilde{P}')$ is usually called the difference class for the bundles $\tilde{P}$ and $\tilde{P}'$. 

\begin{prop}
The difference class is such that $\delta(\tilde{P},\tilde{P}')\cdot \delta(\tilde{P}',\tilde{P}'')=\delta(\tilde{P},\tilde{P}'')$. 
\end{prop}

\begin{proof}
 It is immediate that $\delta(\tilde{P},\tilde{P}')\cdot\delta(\tilde{P}',\tilde{P}'')=\gamma_{\alpha\beta}\gamma_{\alpha\beta}'^{-1}\cdot \gamma_{\alpha\beta}'\gamma_{\alpha\beta}''^{-1}=\gamma_{\alpha\beta}\gamma_{\alpha\beta}''^{-1}=\delta(\tilde{P},\tilde{P}'')$. 
\end{proof}

Now, consider the $\Gamma$-structures $\tilde{P}, \tilde{P}'$ and $\tilde{P}''$ in $P$. The equivalence/non-equivalence between them may be framed as follows: assuming that $\phi:\tilde{P}'\rightarrow\tilde{P}''$ is a $\Gamma$-structures isomorphism and $s'_\alpha:U_\alpha\rightarrow\tilde{P}'$ a local section system, the definition $s''_\alpha=\phi s'_\alpha$ leads to a local section system counterpart to $\tilde{P}''$ (see the diagram below). 

\begin{equation}
\xymatrix{
	\mathcal{M}\supset U_\alpha \ar[rdd]_{s''_\alpha} \ar[rr]^{s'_\alpha} & & \tilde{P}' \ar[ldd]^{\phi}\nonumber \\
	& & \\
	& \tilde{P}''
}
\end{equation} Therefore $s'_\beta=s'_\alpha\gamma'_{\alpha\beta}$ and $s''_\beta=s''_\alpha\gamma''_{\alpha\beta}$. The application of $\phi$ in $s'_\beta$ leads to $s''_\beta=s''_\alpha\gamma'_{\alpha\beta}$, hence $\gamma'_{\alpha\beta}=\gamma''_{\alpha\beta}$. Thus, the difference class $\delta(\tilde{P}',\tilde{P}'')=\gamma'_{\alpha\beta}\gamma_{\alpha\beta}''^{-1}=e$ vanishes. In light of Proposition III.3, this means that $\delta(\tilde{P},\tilde{P}')=\delta(\tilde{P},\tilde{P}'')$. In this way, the equivalence between $\tilde{P}'$ and $\tilde{P}''$ (asserted by the isomorphism) imply $\delta(\tilde{P}',\tilde{P}'')=e$ and $\tilde{P}'=\tilde{P}''$. The reciprocal statement demands a more careful analysis in our case. 

In order to envisage the previous section setup subtleties, let us start assuming now $\delta(\tilde{P}',\tilde{P}'')=e$. This means the very existence of a 1-cochain, $\chi$ $(x\in U_\alpha\cap U_\beta)$ such that\footnote{In fact, notice that $\partial \delta(\alpha,\beta,\kappa)=[\chi(\beta)\chi^{-1}(\kappa)][\chi(\alpha)\chi^{-1}(\kappa)]^{-1}[\chi(\alpha)\chi^{-1}(\beta)]=e$.} $\delta(\alpha,\beta)=\chi(\alpha)\chi^{-1}(\beta)$. Hence, $\gamma'_{\alpha\beta}\gamma_{\alpha\beta}''^{-1}=\chi(\alpha)\chi^{-1}(\beta)$ and, thus $\chi^{-1}(\alpha)\gamma'_{\alpha\beta}=\chi^{-1}(\beta)\gamma''_{\alpha\beta}$. Remembering that $\gamma: U\rightarrow \Gamma$ and $\chi:U\subset M\rightarrow \mathbb{Z}_2\in Z(\Gamma)$, it is possible to write the gamma lifts relation as 
\begin{eqnarray}\label{e0}
\gamma''_{\alpha\beta}=\chi^{-1}(\alpha)\gamma'_{\alpha\beta}\chi(\beta). 
\end{eqnarray} Consider now usual projections $\tilde{\pi}:T\mathcal{M}\rightarrow \mathcal{M}$ and define $W_\alpha=\tilde{\pi}^{-1}(U_\alpha)$. Taking $z\in W_\alpha$, consider also $\gamma_\alpha(z)\in \Gamma$ as the element satisfying\footnote{The element $\gamma_\alpha$ is unique. In fact, if there is $\bar{\bar{\gamma}}_\alpha$ such that  $z=s'_\alpha(\tilde{\pi}'(z))\bar{\bar{\gamma}}_\alpha(z)$, then $s'_\alpha(\tilde{\pi}'(z))\bar{\bar{\gamma}}_\alpha(z)=s'_\alpha(\tilde{\pi}'(z))\gamma_\alpha(z)$, from which the uniqueness is readily obtained.}  $z=s'_\alpha(\tilde{\pi}'(z))\gamma_\alpha(z)$. We are about to define a map between $\Gamma$-structures. However, in light of the last section, this construction must consider the proposed deformation at the fiber bundle level since both maps ($\phi$ and $\bar{\rho}(\xi)$) relate spinor bundles starting from open sets in the base manifold. Therefore, define $\tilde{W}_\alpha:=\mathbb{R}\times W_\alpha$ and $\mathcal{F}$ as a suitable finite functional\footnote{No stringent requirements shall be imputed to $\mathcal{F}.$ It could even be taken as the $\varphi$ function itself.} of $\varphi_\alpha(\tilde{\pi}'(z))$ (the last section deformation function). The mapping $\phi_\alpha:\tilde{W}_\alpha\rightarrow \tilde{P}''$ is defined by
\begin{equation}\label{e1}
\phi_\alpha :=\mathcal{F}[\varphi_\alpha(\tilde{\pi}'(z))]s''_\alpha(\tilde{\pi}'(z))\cdot \chi(\alpha)\gamma_\alpha(z).
\end{equation} Analogously $\phi_\beta =\mathcal{F}[\varphi_\beta(\tilde{\pi}'(z))]s''_\beta(\tilde{\pi}'(z))\cdot \chi(\beta)\gamma_\beta(z)$. Nevertheless, $s''_\beta=s''_\alpha\gamma''_{\alpha\beta}$ and, in view of Eq. (\ref{e0}) $s''_\beta=s''_\alpha\chi^{-1}(\alpha)\gamma'_{\alpha\beta}\chi(\beta)$. Therefore 
\begin{equation}
\phi_\beta(z)=\mathcal{F}[\varphi_\beta(\tilde{\pi}'(z))]s''_\alpha(\tilde{\pi}'(z))\chi^{-1}(\alpha)\gamma'_{\alpha\beta}\chi(\beta)\cdot\chi(\beta)\gamma_\beta(z),\label{e2} 
\end{equation} or, recalling that $\chi\in\mathbb{Z}_2$,  
\begin{equation}
\phi_\beta(z)=\mathcal{F}[\varphi_\beta(\tilde{\pi}'(z))]s''_\alpha(\tilde{\pi}'(z))\chi(\alpha)\gamma'_{\alpha\beta}\gamma_\beta(z).\label{e3} 
\end{equation}

Notice that for $z\in W_\alpha\cap W_\beta$ we have, naturally,
\begin{eqnarray}
z=s'_\beta(\tilde{\pi}'(z))\gamma_\beta(z),\label{e4}\\
z=s'_\alpha(\tilde{\pi}'(z))\gamma_\alpha(z).\label{e5}
\end{eqnarray} From (\ref{e4}) it is straightforward to see that $\gamma'_{\alpha\beta}[s'_\beta(\tilde{\pi}'(z))]^{-1}z=\gamma'_{\alpha\beta}\gamma_\beta(z)$ whose left-hand side amounts out to be 
\begin{eqnarray}\label{e6}
[s'_\beta(\tilde{\pi}'(z))(\gamma'_{\alpha\beta})^{-1}]^{-1}z=[s'_\beta(\tilde{\pi}'(z))\gamma'_{\beta\alpha}]^{-1}z=[s'_\alpha(\tilde{\pi}'(z))]^{-1}z.
\end{eqnarray} Now, Eq. (\ref{e5}) can be used so that Eq. (\ref{e6}) is recognized as being $\gamma_\alpha (z)$ and, hence, $\gamma'_{\alpha\beta}\gamma_\beta(z)=\gamma_\alpha (z)$. Back to Eq. (\ref{e3}) (taking into account Eq. (\ref{e1})) we have 
\begin{equation}
\phi_\beta(z)=\mathcal{F}[\varphi_\beta(\tilde{\pi}'(z))]\mathcal{F}^{-1}[\varphi_\alpha(\tilde{\pi}'(z))]\phi_\alpha(z),\label{e7}
\end{equation} provided $\mathcal{F}$ is non null. This is a relevant result. Firstly, consider regions where $\varphi$ is constant, i. e., $\Sigma_1$ or $\Sigma_3$. Then, Eq. (\ref{e7}) immediately shows that $\phi_\alpha=\phi_\beta$ and, thus, $\phi$ defines a global mapping $\phi:\tilde{P}'\rightarrow \tilde{P}''$. As well known, at the local trivialization level it implies the existence of $q$ such that $qs'_\alpha=qs'_\beta\gamma'_{\beta\alpha}$ or, from Eq. (\ref{e0}), $\chi(\alpha)qs'_\alpha=\chi(\beta)qs_\beta'\gamma''_{\beta\alpha}$, leading to 
\begin{equation}
s''_\alpha=\chi(\alpha)qs'_\alpha.\label{r1}
\end{equation} From Eq. (\ref{e5}) one has $z\cdot \gamma=s'_\alpha(\tilde{\pi}'(z\cdot\gamma))\gamma_\alpha(z\cdot\gamma)$. On the other hand, $z\cdot\gamma=s'_\alpha(\tilde{\pi}'(z))\gamma_\alpha(z)\cdot \gamma$. Therefore, $s'_\alpha(\tilde{\pi}'(z\cdot\gamma))\gamma_\alpha(z\cdot\gamma)=s'_\alpha(\tilde{\pi}'(z))\gamma_\alpha(z)\cdot\gamma$ and entering $\chi(\alpha)q$ from the left we get, with the aid of Eq. (\ref{r1})
\begin{equation}
s''_\alpha(\tilde{\pi}'(z\cdot\gamma))\gamma(z\cdot\gamma)=s''_\alpha(\tilde{\pi}'(z))\gamma_\alpha(z)\cdot\gamma. \label{r2}
\end{equation} Remembering, once again, that $\chi\in\mathbb{Z}_2\subset Z(\Gamma)$ and bearing in mind Eq. (\ref{e1}) for $\varphi$ constant, we finally arrive at $\phi(z\cdot\gamma)=\phi(z)\cdot\gamma$ and the isomorphism between the $\Gamma$-structures is complete. Hence, for the aforementioned regions, we have recovered the standard result, namely, $\tilde{P}'\simeq\tilde{P}''$ if, and only if, $\delta(\tilde{P}',\tilde{P}'')=e$. 

The novelty of the presented construction rests upon the region $^\varphi\mathbb{R}^{1,3}$. Here, even being a region of trivial topology (in the sense that $\pi_1(^\varphi\mathbb{R}^{1,3})=0$), Eq. (\ref{e7}) shows that the $\varphi$-deformation prevents the existence of a global mapping between $\tilde{P}'$ and $\tilde{P}''$. In other words, exotic spinors are indeed {\it nonequivalent} to standard ones, having dynamics corrected by a term taking the deformation into account, regardless of the mentioned constancy of $\theta(x)$ alluded to in the last section. Of course, the mapping may well exist locally, but not globally in $^\varphi\mathbb{R}^{1,3}$. So, the situation in $\mathcal{M}$ is as follows: since $\pi_1(\mathcal{M})\neq 0$, the first \v{C}ech cohomology group $\check{H}^1(\mathcal{M},\mathbb{Z}_2)\neq 0$ and $\mathcal{M}$ will have as many nonequivalent spinor structures as elements in $\check{H}^1$. In the light of the previous section setup, exotic spinors are distinguishable from usual ones in $\Sigma_1$ and, due to the $\varphi$-deformation, also in $^\varphi\mathbb{R}^{1,3}$. In the next section, we shall discuss criteria for possible topology coexisting, so to speak, in $\mathcal{M}$.
 
\section{physical consequences}

We split this section into three parts. The first one aims to discuss the current conservation further in light of the proposed separation between spacetime regions. The second one, accounting for a quasinormal-like behavior and an application of the Gordon decomposition, and the third one, where we extend the results presented in Ref. \cite{np}, where an approach to geometrical consequences of nontrivial topology was settled by investigating a nonlinear sigma model. 

\subsection{Junction condition}

The separation in spacetime regions we made before forces us to a more precise notion of current conservation. We shall deal exclusively with flat spacetime. Nevertheless, for the sake of generality, we develop a more comprehensive analysis in this subsection. Based on the matching conditions developed in Ref. \cite{mcf}, consider an orientable spacelike hypersurface $\Delta\subset \mathcal{M}$. We shall also open the possibility of a timelike hypersurface. Consider a congruence of geodesics piercing $\Delta$ orthogonally and denote the proper distance coordinates along the congruence by $r\in\mathbb{R}$, parameterized such that $r$ vanishes at $\Delta$. In the case of a timelike $\Delta$, $r$ stands for the proper time along the congruence. This construction allows to set $r\in \mathbb{R}^*_-$ denoting one side of $\Delta$, while $r\in \mathbb{R}^*_+$ denotes its side. As the hypersurface $\Delta$ is assumed orientable, we are also allowed to define $\{n_\mu\}$ as the set of vectors normal to $\Delta$ at every point, in view of what infinitesimal displacements away from $\Delta$ are settled by $dx^\mu=n^\mu dr$. There are two immediate consequences: the first one is that $dx^\mu dx_\mu=n^\mu n_\mu (dr)^2$ and therefore $n^\mu n_\mu>0$ for spacelike $\Delta$ ($n^\mu n_\mu<0$ for timelike $\Delta$); the second point is that $n_\mu dx^\mu=n_\mu n^\mu dr$ and, therefore, $n_\mu=\pm\partial_\mu r$ by properly normalizing $n_\mu$. Recall the standard definition of the Heaviside distribution $\Theta(r)$: $\Theta(r)=1$, if $r>0$, $\Theta(r)=0$, if $r<0$, and indeterminate if $r=0$. 

As a last ingredient, define a continuous and differentiable coordinate system, $\{x^\mu\}$, traversing both sides of $\Delta$. The current decomposition (thought as a distribution-valued vector) on either side of $\Delta$ reads
\begin{equation} 
j^\mu=\Theta(r)j^\mu_+ +\Theta(-r)j^\mu_-,
\end{equation} where $j^\mu_+$ ($j^\mu_-$) denotes the current on the $r>0$ ($r<0$) side of $\Delta$. Taking the derivative leads to 
\begin{equation} \label{sing}
\partial_\mu j^\mu=\Theta(r)\partial_\mu j^\mu_+ +\Theta(-r)\partial_\mu j^\mu_- \pm\delta(r)(j^\mu_+ -j^\mu_-)n_\mu,
\end{equation} where use was made of the Dirac delta distribution resulting from the Heaviside derivation and, in the last term, the upper (down) sign stands for spacelike (timelike) $\Delta$. The current conservation follows straightforwardly from the discussion presented at the final of Section II, provided we made $\Delta$ coincide with the interface between $\Sigma$ regions, use the corresponding dynamical equation for respective $\Delta$ side, and require the vanishing of the singular term in (\ref{sing}). The current is conserved everywhere (with conservation law existing as a distribution) if, and only if, $j^\mu$ is continuous across the hypersurface $\Delta$. The behavior already required for the deformation function fills this constraint. 

Finally, denote $\{z^m\}$ a coordinate system within $\Delta$ and parameterize it as $x^\mu=x^\mu(z^m)$. Let $C$ be a set of curves belonging to $\Delta$. Hence, vectors tangent to an element of $C$ are given by $e^\mu_{\,m}=\partial x^\mu/\partial z^m$ and, naturally, $e^\mu_{\,m} n_\mu=0$. Besides, continuity of $\{x^\mu\}$ ensures $(e_+)^{\mu}_{\;m}=(e_-)^{\mu}_{\;m}$. Therefore, contraction of the conservation condition $j^\mu_+=j^\mu_-$ with $e^m_{\;\mu}$ gives $j^m_+=j^m_-$, and the induced current must be the same on both sides of $\Delta$. This is the junction condition in our case.  

\subsection{Kinematic and dynamic effects}

From now on, we shall deal with the physical consequences of exotic spinors in the $^\varphi\mathbb{R}^{1,3}$ region (and a residual effect in $\Sigma_3$). As discussed in the preceding sections, calling $\varphi^{-1}\partial_\mu\varphi\equiv k_\mu$, the equation of motion reads $i\gamma^\mu(\partial_\mu+k_\mu)\tilde{\psi}-m\tilde{\psi}=0$. The $k_\mu$ vector results from the deformation procedure, which, by its turn, is possible due to the $\tilde{P}\rightarrow P$ mapping. It then encodes the nontrivial topology echo in $^\varphi\mathbb{R}^{1,3}$. We set as appropriate conditions for this deformation an approximate constant vector $k_\mu$ of small magnitude so that $\partial k$ and $k^2$ shall be disregarded from the physical analysis. These conditions are the practical implementation of the mathematical slowly varying function discussion found in the appendix.

The presence of $i\gamma^\mu k_\mu$ in the dynamical equation may be regarded as an optical potential-like term. In order to see this, let us investigate the dispersion relation evaluating $[i\gamma^\nu(\partial_\nu+k_\nu)+m][i\gamma^\mu(\partial_\mu+k_\mu)-m]\tilde{\psi}=0$. A bit of straightforward algebra leads to 
\begin{eqnarray}
\Box\tilde{\psi}+m^2\tilde{\psi}+2k^\mu\partial_\mu\tilde{\psi}+\gamma^\mu\gamma^\nu\partial_\mu k_\nu\tilde{\psi}+k^2\tilde{\psi}=0,
\end{eqnarray} which, after the imposition of the above simplifications, reads simply 
\begin{eqnarray}
\Box\tilde{\psi}+m^2\tilde{\psi}+2k^\mu\partial_\mu\tilde{\psi}=0.\label{aaa1}
\end{eqnarray} The first derivative term is the optical potential analog and shall lead to quasinormal-like modes to the exotic fermion. Quasinormal modes are indeed valuable for treating the dissipative feature of several systems \cite{CE}, from black holes \cite{quasi} to condensed matter physics \cite{CM}. To evince this behavior in our context, consider a region $\Omega\subset \,^\varphi\mathbb{R}^{1,3}$ whose boundary will be denoted by $\partial\Omega$. The spinorial Fourier transform of (\ref{aaa1}) gives 
\begin{eqnarray}   
\int_\Omega d^4x (-p^2+m^2-2ik^\alpha p_\alpha)\tilde{\psi} \exp{(ipx)}+\bigg\{[\partial^\alpha\tilde{\psi}+(2k^\alpha-ip^\alpha)\tilde{\psi}]\exp{(ipx)}\bigg\}_{\partial\Omega}=0.
\end{eqnarray} The imposition of Neumann-Dirichlet boundary conditions allows one to write 
\begin{equation}
E^2+2ik^0E-{\bf p}\cdot({\bf p}+2i{\bf k})-m^2=0,
\end{equation} giving, in the approximation scope,
\begin{eqnarray}\label{aaa2}
E\simeq \pm\sqrt{{\bf p}^2+m^2}+i\Bigg\{-k^0\pm\frac{{\bf p}\cdot{\bf k}}{\sqrt{{\bf p}^2+m^2}}\Bigg\},
\end{eqnarray} from which the quasinormal-like behavior is immediate. Accommodating what is expected for an observer in rest with respect to the particle, one should, quite prudently, require $k^0=0$, that is\footnote{The vanishing of $k^0$ is not mandatory (at least not for this reason), since the rest observer with respect to the particle is not inertial. Nevertheless, even for $k^0=0$, as as it is clear, the quasinormal-like behavior is in order. We shall assume $k^0=0$ for simplicity.}  $\varphi=\varphi(\bf{x})$. The first derivative term is a consequence of the deformation taken, by its turn, as an echo of nontrivial topology. As we can see from Eq. (\ref{aaa2}), the dispersion relation is impacted so that usual pure oscillating terms are replaced by $\exp{[i(Re(p_0)+iIm(p_0))t]}$ and an overall friction-like factor $\exp{[\mp {\bf p}\cdot {\bf k}/\sqrt{{\bf p}^2+m^2}]}$ appears. The term `friction' was somewhat loosely used, but it can be upgraded with a more concrete sense as follows: the constant vector ${\bf k}=\varphi^{-1}{\bf\nabla}\varphi$ encodes the deformation function gradient. For states approaching the $\mathcal{V}\subseteq\Omega$ region, it seems natural to think of an outgoing gradient as a barrier type. In contrast, an ongoing gradient should drag field modes to $\mathcal{V}$ (with a similar, but opposite, interpretation for states leaving the region). Therefore, the field modes momentum is anti-parallel to ${\bf k}$, i.e., ${\bf p}\cdot{\bf k}<0$, rendering a truly friction term damping approaching modes, and reinforcing leaving modes. Given the approximations, one can expect that the damping and reinforcing magnitudes are similar but do not compensate for each other quite exactly. 

We also remark that from the dispersion equation, it is also possible to find out the group velocity $v_j=\partial E/\partial p_j$ (in components) given by 
\begin{eqnarray}\label{aaa3}
v_j=\pm \frac{p_j}{\sqrt{{\bf p}^2+m^2}}\mp \frac{i}{\sqrt{{\bf p}^2+m^2}}\Bigg[\frac{({\bf p}\cdot{\bf k})p_j}{{\bf p}^2+m^2}-k_j\Bigg],
\end{eqnarray} corroborating the quasinormal-like behavior. Both expressions, (\ref{aaa2}) and (\ref{aaa3}), reduce to the standard ones in the limit $k_\mu\rightarrow 0$, as expected. Besides, always in the discussed approximation scope, $\eta^{jk}v_j v^* _k=|{\bf v}|^2<1$, indicating the conservation of causal structure in $\mathcal{V}$.      

So far, we have dealt with the impacts of the Lorentz violation on the wave function. It is instructive to extend the analysis by looking at Gordon's current decomposition. In doing so, we shall work at $\Sigma_3$ and face a consequence of the residual constant $\varphi$ aforementioned in Section II. We start performing the calculation in general, considering $k_\alpha$ and particularizing to a region of $\Sigma_3$ afterward. From the equation of motion, we have 
\begin{equation}
\bar{\tilde{\psi}}\gamma^\mu(m\tilde{\psi})=i\bar{\tilde{\psi}}\gamma^\mu\gamma^\alpha\partial_\alpha\tilde{\psi}+i\bar{\tilde{\psi}}\gamma^\mu\gamma^\alpha k_\alpha\tilde{\psi}, \label{ab1}
\end{equation} while     
\begin{equation}
(\bar{\tilde{\psi}}m)\gamma^\mu\tilde{\psi}=-i\partial_\alpha\bar{\tilde{\psi}}\gamma^\alpha\gamma^\mu\tilde{\psi}-ik_\alpha\bar{\tilde{\psi}}\gamma^\alpha\gamma^\mu\tilde{\psi}.\label{ab2}
\end{equation} Adding the two equations above, and using the standard gamma matrices relation $\gamma^\mu\gamma^\nu=\eta^{\mu\nu}-i\sigma^{\mu\nu}$, where $\sigma^{\mu\nu}=\frac{i}{2}[\gamma^\mu,\gamma^\nu]$, it is possible to achieve a factored expression for $j^\mu=\varphi^2\bar{\tilde{\psi}}\gamma^\mu\tilde{\psi}$ given by  
\begin{eqnarray}
j^\mu=\frac{i\varphi^2}{2m}(\bar{\tilde{\psi}}\partial^\mu\tilde{\psi}-\partial^\mu\bar{\tilde{\psi}}\tilde{\psi})+\frac{\varphi^2}{2m}\bigg\{\partial_\alpha(\bar{\tilde{\psi}}\sigma^{\mu\nu}\tilde{\psi})+2k_\alpha\bar{\tilde{\psi}}\sigma^{\mu\alpha}\tilde{\psi}\bigg\}\label{ab3}
\end{eqnarray} and recalling that $k_\alpha=\varphi^{-1}\partial_\alpha\varphi$, Eq. (\ref{ab3}) can be recast  
\begin{eqnarray}
j^\mu=\frac{i\varphi^2}{2m}(\bar{\tilde{\psi}}\partial^\mu\tilde{\psi}-\partial^\mu\bar{\tilde{\psi}}\tilde{\psi})+\frac{1}{2m}\partial_\alpha\big(\varphi^2\bar{\tilde{\psi}}\sigma^{\mu\alpha}\tilde{\psi})\big)\equiv \varphi^2 j^\mu_1+j^\mu_2.\label{ab4}
\end{eqnarray} This last expression resembles the usual Gordon decomposition, except by $\varphi$ terms. The first part of Eq. (\ref{ab4}) is less relevant to our purposes. The exotic current coupling to the electromagnetic field is achieved by the Hamiltonian density $H_{int}=-j^\mu A_\mu$ from which we select $-j^\mu_2A_\mu$ for further analyses. To properly consider this coupling, we shall take some care: the first point is that such a coupling is usually taken concerning the (electric) charge $e$. Here we are assuming that the $A_\mu$ field is coupling to the charge coming from the conservation of $j^\mu=\varphi^2\bar{\tilde{\psi}}\gamma^\mu\tilde{\psi}$.  If one takes this coupling to {\it bona fide}, then it is clear that its consequences shall also remain in the formal limit $e\rightarrow 0$. Also, we now concentrate the analysis to $\Omega_3\subset\Sigma_3$. The reason for that is the possibility of thinking of the standard quatization procedure and nonrelativistic limit on a usual basis, the only peculiarity coming from a constant value for $\varphi\equiv \bar{\varphi}$.    

Given the above discussion, we have 
\begin{eqnarray}
H_{int}\supset -\frac{1}{2m}\partial_\alpha(\bar{\varphi}^2\bar{\tilde{\psi}}\sigma^{\mu\alpha}\tilde{\psi}A_\mu)+\frac{1}{2m}\bigg(\frac{1}{2}\bar{\varphi}^2 F_{\alpha\mu}\bar{\tilde{\psi}}\sigma^{\mu\alpha}\tilde{\psi}\bigg),\label{ab5}
\end{eqnarray} where $F_{\mu\nu}$ is the electromagnetic field strength. The first term will be disregarded upon volume integration in $\Omega_3$. In contrast, for the second term, appropriate textbook nonrelativistic limit \cite{PS} leads to paradigmatic magnetic moment factor, in our case given by $\mu=\frac{-\bar{\varphi}^2}{2m}$. This is relevant since neutral particles are known to have sizable magnetic moments, and no charge is presented in our formulation. For instance, this fact can be used to explore the neutron magnetic moment with an effective coupling automatically raised using exotic deformed neutral fermions. In this case, we can make contact with experimental values, remembering that $\mu_N=|e|/2m_p$, where $m_p$ is the proton mass and hence $\mu_N$ stands for the nuclear magneton. Thus $\mu/\mu_N=-\bar{\varphi}^2m_p/(|e|m)$. For the neutron, the rate $\mu/\mu_N$ is about\footnote{The values are here presented to furnish an estimate, without much care about significance, errors, and so on.} $-1.913$ \cite{cod}. Therefore, setting for a moment the exotic spinor mass as $m=\beta m_p$ (where $\beta$ is a real constant), we find $\bar{\varphi}^2/\beta=1.913|e|$. From the fine structure constant (in natural units so chosen that the vacuum electric permittivity equals $1$), we have $e^2/4\pi=1/137$, enabling one to take $|e|\simeq 0.303$. This leads to $\bar{\varphi}^2/\beta\simeq 0.5796$, and by requiring (for self-consistency) that the exotic fermion mass equals the neutron mass ($m=m_n$), we arrive at $\bar{\varphi}^2\simeq 0.5796\frac{m_n}{m_p}$. This fixes the residual value for $\bar{\varphi}$ at about $0.7618$, a value within the allowed range for the deformation.

Some remarks are necessary for this last application. Firstly, as it is clear, we are completely ignoring any particle internal structure and only considering an effective fermionic field \cite{ictp}. We are also not claiming that exotic spinors should describe neutral fermions, but instead saying that, in doing so, a complex result can be made straightforward in this effective approach. As we made clear, we particularized the final part of this analysis within $\Sigma_3$ to use the usual equation of motion and make contact with the standard nonrelativistic limit. Had we insisted in work in $^\varphi\mathbb{R}^{1,3}$, it is not clear which contribution should come from $\varphi^{-1}\partial\varphi$. It seems, however, that due to the deformation function's slowly varying behavior, one could expect a running (or, better said, a walking) coupling, apparently leading to a varying magnetic moment. It is a genuine question whether this controlled Lorentz-violation approach can reduce the theoretical/experimental tension of the muon $g-2$ factor. While it certainly requires a more careful analysis, we are inclined to say that an investigation that parallels what we have done here would suffer from similar problems of other Lorentz-violation models \cite{jjh}.    
 
It is important to say a word about the $\bar{\varphi}$ residual term. Firstly, as it can be seen, denoting by $d\bf{l}$ an infinitesimal displacement along a curve $\alpha\subset \Sigma_2$, it is fairly simple to see that ${\bf k}\cdot d{\bf l}=\varphi^{-1}d\varphi$, from which we have $\varphi({\bf x})=\varphi_0\exp{({\bf k}\cdot\int^{{\bf x}}d{\bf l})}$, where $\varphi_0$ is a constant and the integral is taken over a smooth curve $\alpha\subset \Sigma_2$. In this vein, the diameter of $\Sigma_2$ is related to the saturation value of $\varphi\equiv \bar{\varphi}$. On a different basis, as pointed out in Ref. \cite{np}, the impact coming from the terms correcting (and deforming) the spinor derivative may be interpreted in a less literal basis than the space-time separation proposed. In natural units, these terms scale with energy. Grand energy scales (and, equivalently, short distances scrutinized) bring more significant effects from the deformation term. Therefore, one could interpret the $\Sigma$ regions as down-deep length scale regimes. In this regard, the residual term could be framed as a low-energy vestige of high-energy Lorentz breaking, possibly due to the deformation taken upon the nontrivial topology. Still, there is no saturation mechanism for the deformation, but only a fit with a physical effect. So, any serious tentativeness in following this interpretation should consider that we may be pushing our results too forward by studying higher-energy field theoretical phenomena via a single-particle wave function.   

It can be seen that while the standard interpretation of the analyzed setup shows a fixed topology, the field propagating modes, when going through $\Sigma_1$ to $\Sigma_2$ regions (or vice-versa), experience a topology change. It is unclear whether all the issues of dynamical topology are necessarily imputed to our approach. However, it is possible to depict a picture in the context studied. As analyzed in Ref. \cite{Anderson:1986ww} and followed by Ref. \cite{trousers_revisited}, in such a context, particle creation related to the nontrivial effect in the field modes due to the spacetime topology change is expected. This effect is described by two different vacua related via a Bogoliubov transformation. However, infinitely many particles may be created depending on the topology change.

Let us analyze the case of a massless scalar field whose asymptotic modes dynamic is given in $\mathbb{R}^4$ for $t<0$ and $\mathbb{R}\times(\mathbb{R}^2\times S^1)$ for $t>0$. The momentum field modes associated with $S^1$ are expected to assume a discrete distribution behavior as (multiply integers of) $R^{-1}$, where $R$ is a typical scale associated with the $S^1$ radius. This partial discretization also impacts the frequency, and divergent behavior occurs at high energy scales \cite{pd} (see also \cite{dowker} for a complementary approach). Such a behavior is also expected (though less severe) for Dirac fields and also for the opposite case, i.e., asymptotic topology going as $\mathbb{R}\times(\mathbb{R}^2\times S^1)$ for $t<0$ and $\mathbb{R}^4$ for $t>0$. The approaches to deal with these divergences can be brought into three main classes: 1) the usual introduction of cutoffs to regulate the divergence \cite{zeld,fulling}, turning the theory adequate up to an effective energy scale (in consonance with which we framed our discussion so far); 2) the use of twisted boundary conditions \cite{FV}, in which appropriate boundary conditions changes the spectrum in a controlled fashion (in this case, it is necessary additional care to see whether the effect of the suitable boundary conditions is preserved when putting the fields in a box for quantization - see discussion below); 3) considerations about smooth topology changes employing adiabatic regularization technics \cite{BD}, in which the radius $R_{\epsilon}(t)$ is a suitable smooth function of the time coordinate multiplied by the $\epsilon$ (with dimension $length^{-1}$) parameter. Considering $\epsilon\rightarrow 0$ as an appropriate limit ensures a slow enough compactification (or decompactification), controlling the divergences. While lacking a mechanism, this last possibility is particularly appealing in our discussion since it mimics the deformation function of slow behavior at the topology level.

It is a good point to emphasize some technical aspects. Firstly, not all topologies can be considered for the sigma regions. A conservative proposal goes as follows: consider $\Sigma^C_i$ $(i=1,2)$ a closed (compact without boundary) spacelike hypersurface contained in $\Sigma_i$. It is formally the physical construction equivalent of putting the above fields in a box (with appropriate boundary conditions). Let $\Sigma^C_1 \uplus \Sigma^C_2$ denotes the disjoint union. Take $\mathcal{M}^C\subset\mathcal{M}$ a smooth compact four-dimensional manifold whose boundary is $\partial \mathcal{M}^C=\Sigma^C_1 \uplus \Sigma^C_2$. As $\mathcal{M}^C$ is endowed with a Lorentzian metric, the above structure is called a Lorentzian cobordism, with $\mathcal{M}^C$ smoothly interpolating between $\Sigma^C_i$'s. The existence of such a cobordism is usually used as a criterion in allowed quantum topology changes \cite{sor}. We believe it can also be helpful in our case, defining clear-cut criteria along with explicit drawbacks. It is shown in Ref. \cite{novomilnor} that a necessary and sufficient condition to the existence of a cobordism is the equality of the Stiefel-Whitney numbers of $\Sigma^C_i$'s. These are the topological criteria for allowed transitions experienced by the physical fields in our context. As it can be seen, the Stiefel-Whitney numbers of $\Sigma_1$ and $\Sigma_2$ are the same, and we shall expect the same thing to happen for $\Sigma^C_i$'s. The other technical points follow from the considerations we have just made. In general, Lorentzian cobordisms imply closed timelike curves \cite{outger}, which can be avoided by allowing the metric to degenerate at some points \cite{kin}. This last point, however, is certainly an aspect that lies beyond the scope of this paper.

Finally, different (noncompact or with boundary) cobordisms could be tried to define topological criteria for topology change, and even more radical approaches where cobordism is not essential can also be tried. These possibilities deserve attention in the future.

\subsection{Equations of motion of a nonlinear sigma model built upon a slightly deformed metric}

This subsection will discuss a usual nonlinear sigma model based upon a metric-deformed spacetime. The idea further explores the construction performed in Ref. \cite{np}. It is outside the main line of the paper but may also be faced as an application, though in a different (`geometrized') context. The target space is $SO(N+1)/SO(N)$, and the metric deformation is the one developed in Ref. \cite{np}. Since most of the ideas were developed in detail in \cite{np}, here we only introduce the main steps to arrive at the deformation.   

We start recalling that a given spacetime point $P\in \Sigma_3$ can be written in such a way that its coordinates are given in terms of spinor entries \cite{penro}, $x^\mu \thicksim (\zeta\zeta^*)^\mu$, where $\zeta$ is a spinor entry and $\zeta^*$ is its complex conjugate. Again, we stress that the formal details may be found in \cite{np} and it suffices to our purposes saying that $dx^\mu=d(\zeta\zeta^*)^\mu=\frac{\partial (\zeta\zeta^*)^\mu}{\partial_\nu}dx^\nu=(\partial_\nu\zeta\zeta^*+\zeta\partial_\nu\zeta^*)^\mu dx^\nu$. In $\Sigma_3$, the derivation process can now be taken backward, and no effect is expected. It is unclear whether the formalism of writing spacetime points as spinor entries condensate, so to speak, can be immediately extended to the other regions contemplated in this paper. In any case, given the procedure's validity in $\Sigma_3$, we now analyze its consequences as the $^\varphi\mathbb{R}^{1,3}$ region is approached. This reasoning leads to the $\varphi$-deformation `geometrization' possibility. Recall the discussion around Eq. (\ref{cor2}): the matrix part of the $Spin(1,3)$ mapping is just $\mathbb{1}_{M(4,\mathbb{C})}$ and, therefore, the mapping between $\tilde{P}$ and $P$ is taken entry to entry, just as the derivative deformation. That is to say, in $^\varphi\mathbb{R}^{1,3}$ we may write 
\begin{equation}
dx^\mu\mapsto \{(\partial_\nu\zeta+\varphi^{-1}\partial_\nu\varphi\zeta)\zeta^*+\zeta(\partial_\nu\zeta^*+\varphi^{-1}\partial_\nu\varphi\zeta^*)\}^\mu dx^\nu=dx^\mu+2x^\mu(\varphi^{-1}\partial_\nu\varphi) dx^\nu.\label{b1}
\end{equation} This replacement engenders a deformation in the $(^\varphi\mathbb{R}^{1,3})^*$ basis, which impacts differential form coefficients. In this vein, consider $\tilde{\eta}\in(^\varphi\mathbb{R}^{1,3})^*\otimes(^\varphi\mathbb{R}^{1,3})^*$: 
\begin{equation}  
\tilde{\eta}=\eta_{\mu\nu}(dx^\mu+2x^\mu\varphi^{-1}\partial_\alpha\varphi dx^\alpha)\otimes(dx^\nu+2x^\nu\varphi^{-1}\partial_\beta\varphi dx^\beta)
\end{equation} and hence, acting upon vectors $\{e_\mu\}$ (base elements of $^\varphi\mathbb{R}^{1,3}$) we have, after a bit of algebra\footnote{Taking into account the $\varphi$ function approximations dealt with in last Section.}, 
\begin{equation}
\tilde{\eta}(e_\mu,e_\nu)=\eta_{\mu\nu}+2\eta_{\mu\alpha}x^\alpha\varphi^{-1}\partial_\nu\varphi+2\eta_{\nu\alpha}x^\alpha\varphi^{-1}\partial_\mu\varphi\equiv \tilde{\eta}_{\mu\nu}.
\end{equation}

Now, we are in a position to analyze the nonlinear sigma model. As mentioned, take the target space as being the $N-$dimensional unitary sphere $S^N$ by means of $\{\Phi^I(x)\}$ fields set ($I=1,\cdots,N,N+1\equiv i,N+1$) respecting the $\Phi^I(x)\Phi^I(x)=1$ constraint. This enables a mapping $^\varphi\mathbb{R}^{1,3}\rightarrow S^N$ ($x^\mu\mapsto \Phi^i(x)$) and the free lagrangian $(1/2)\tilde{\eta}_{\mu\nu}\partial^\mu\Phi^I\partial^\nu\Phi^I$ gives rise to the sigma model (after imposition of the constraint) 
\begin{eqnarray} 
\mathcal{L}_\sigma=\frac{1}{2}\tilde{\eta}_{\mu\nu}g_{ij}(\Phi)\partial^\mu\Phi^i\partial^\nu\Phi^j,
\end{eqnarray} where $g_{ij}(\Phi)=\delta^{ij}+\Phi^i\Phi^j/(1-\Phi^k\Phi^k)$. The equations of motion can be obtained from the $\mathcal{L}_\sigma$ functional variation with respect to the $\Phi$ field. Thus, denoting $\partial_a=\partial/\partial\Phi^a$, we have  
\begin{eqnarray}
\delta\mathcal{L}_\sigma=\frac{1}{2}\tilde{\eta}_{\mu\nu}(\partial_k g_{ij}(\Phi))\delta\Phi^k\partial^\mu\Phi^i\partial^\nu\Phi^j+\tilde{\eta}_{\mu\nu} g_{ij}(\Phi)\partial^\mu(\delta\Phi^i)\partial^\nu\Phi^j.\label{b2}
\end{eqnarray} The second term may be recast 
\begin{eqnarray}\label{b3}
\tilde{\eta}_{\mu\nu}g_{ij}(\Phi)\partial^\mu(\delta\Phi^i)\partial^\nu\Phi^j&=&\left.\partial^\mu(\tilde{\eta}_{\mu\nu}g_{ij}(\Phi)\delta\Phi^i\partial^\nu\Phi^j)-10\varphi^{-1}\partial_\nu\varphi g_{ij}(\Phi)\partial^\nu\Phi^j\delta\Phi^i\right.\nonumber\\&-&\left.\tilde{\eta}_{\mu\nu}\partial^\nu\Phi^j\partial_k g_{ij}(\Phi)\partial^\mu\Phi^k\delta\Phi^i-\tilde{\eta}_{\mu\nu}g_{ij}(\Phi)\partial^\mu\partial^\nu\Phi^j\delta\Phi^i.\right.
\end{eqnarray} Disregarding the total derivative term and reinserting Eq. (\ref{b3}) back to (\ref{b2}), the equations of motion are obtained from $\delta\mathcal{L}_\sigma/\delta \Phi^i=0$. It can be readily verified that they read 
\begin{equation}
\tilde{\eta}_{\mu\nu}\partial^\mu\partial^\nu\Phi^p+\tilde{\eta}_{\mu\nu}\Gamma^p_{ij}(\Phi)\partial^\mu\Phi^i\partial^\nu\Phi^j+10k_\alpha\partial^\alpha\Phi^p=0,\label{b4}
\end{equation} where $\Gamma^p_{ij}(\Phi)$ stands for the (second kind) Christoffel symbols in terms of $g_{ij}(\Phi)$ and $k_\alpha$ is defined as before. In sequel, defining $A^{i\;\;\mu}_{\;k}(\Phi)\equiv \Gamma^i_{jk}(\Phi)\partial^\mu\Phi^j$ and $D^{i\;\;\mu}_{\;k}\equiv \delta^i_k\partial^\mu+A^{i\;\;\mu}_{\;k}(\Phi)$, the equations of motion can be written compactly as 
\begin{equation}
(\tilde{\eta}_{\mu\nu}D^{i\;\;\mu}_{\;k}+10k_\nu\delta^i_k)\partial^\nu\Phi^k=0,\label{b5}
\end{equation} suggesting the introduction of the `covariant derivative' $\nabla^{i\;\;\mu}_{\;k}=\tilde{\eta}_{\mu\nu}D^{i\;\;\mu}_{\;k}+10k_\nu\delta^i_k$. Of course, when $k_\alpha\rightarrow 0$ ($\tilde{\eta}_{\mu\nu}\rightarrow \eta_{\mu\nu}$), we recover the usual nonlinear sigma model covariant derivative operator, as expected, from which a target space Riemannian curvature may be read.   

\section{final remarks}

%INTERPRETAÇÃO EM TERMOS DE ESCALA DE ENERGIA / CRITICA A ESSA INTERPRETACAO / [D\Psi] E FUNCIONAL GERADOR / nervo e limite direto?

Following the previous results on the geometrization of topology in \cite{np}, we further investigate how deformations in spinor bundles, induced by the underlying topology of the base manifold, affect the physical properties of spinor fields. The behavior of the deformation function $\varphi(x)$ allowed us to model regions of the manifold where Lorentz and Poincar\`e symmetries are violated, leading to the emergence of unique physical characteristics.

One of the critical aspects of this study is the realization that the deformation function $\varphi(x)$ may creates distinct physical regions within the manifold, each with its own set of properties. In the region $\Sigma_1$, where the topology is nontrivial, the spinor fields exhibit behavior that is markedly different from that in regions $\Sigma_2$ and $\Sigma_3$. The exotic spinors in $\Sigma_1$ arise from a nontrivial spinor structure that cannot be smoothly connected to the standard spinor structure found in trivial topological regions. This profoundly impacts the physical observables associated with the spinor fields.

In the intermediate region $\Sigma_2$, where the deformation function $\varphi(x)$ is non-constant, we observe that the deformation influences the spinor fields in a way that violates Lorentz symmetry. The magnitude of this violation is directly related to the gradient of $\varphi(x)$, which acts as a source of symmetry-breaking. This finding is significant as it provides a new and concrete example of how topological features of the manifold can lead to observable physical effects. Furthermore, the analysis of the current conservation and the dynamical equations for the spinor fields revealed that the deformation induced by $\varphi(x)$ introduces additional terms that manifest as optical potential-like terms, which can lead to quasinormal modes in the spinor field solutions. The presence of such modes suggests that the deformed spinor fields exhibit dissipative behavior, a feature that is typically associated with systems losing energy, such as perturbed black holes or other dissipative physical systems.

The final region $\Sigma_3$, characterized by constant $\varphi(x)$, serves as a baseline for understanding the effects of the deformation. In this region, the spinor fields revert to their standard form, and the physical consequences of the deformation are minimized. However, even in this region, the residual effects of the deformation can be observed. The analysis of the magnetic moment via exotic spinors in $\Sigma_3$ revealed that the deformation function $\varphi(x)$ plays a crucial role in determining the magnitude of the magnetic moment, with potential implications for the study of neutral particles and their interactions with electromagnetic fields.

In conclusion, this paper has demonstrated that deformations in spinor dynamics, driven by the topology of the underlying manifold and the introduction of the deformation function $\varphi(x)$, have significant and far-reaching consequences for the physical properties of spinor fields. The results contribute to enlightening the understanding of the interplay between topology, symmetry, and geometry in theoretical high-energy physics. Future research could expand these findings by exploring more physical scenarios where spinors play a major role or comparing the results to other Lorentz violation scenarios.

\section*{Acknowledgments}

JMHS thanks to CNPq (grant No. 307641/2022-8) for financial support and GMCR thanks to CAPES for financial support.  

\section*{Appendix: A useful associated Hausdorff measure}

This appendix is devoted to studying a measure of the local deformation degree in $^\varphi\mathbb{R}^{1,3}$ aimed to arrive at a slowly varying $\varphi$-deformation criteria compatible with the idea of $^\varphi\mathbb{R}^{1,3}$ as an intermediate region. To avoid complications, consider the Euclidean\footnote{We shall not be concerned with spinor bundles in this appendix.} version of $^\varphi\mathbb{R}^{1,3}$, say $^\varphi\mathbb{R}^{4}$, in which the concept of open balls is intuitive. Since the region of interest is a part of the whole base manifold, the very concept of the Hausdorff measure furnishes an appropriate tool. We shall start recalling a few definitions \cite{roy,elon} leading to the concept of the Hausdorff measure.  

\begin{definition}
	Let $K\subset \mathbb{R}^n$. The diameter of $K$, $|K|$, is the supremum of distances between two of their points:
	\begin{equation}
	|K|=\sup \{||x-y||\mid x,y\in K\}.
	\end{equation} 
\end{definition}	

\begin{definition}
	A $\delta$-covering for $K\subset \mathbb{R}^n$ is an enumerable collection of sets $\{x_\gamma\}$ with maximum diameter $\delta$ covering $K$, i.e.,
	\begin{equation}
	\{x_\gamma\mid K\subset \cup_{\gamma=1}^\infty x_\gamma \;\, and \;\, 0<|x_\gamma|\leq\delta \;\forall \gamma \}.
	\end{equation}
\end{definition} Consider $s\geq 0$ and $K\subset \mathbb{R}^n$. The following definition arises in searching for coverings of $K$ restricted to small values for $\delta$.

\begin{definition}
	The $s$-dimensional Hausdorff measure of $K$ reads 
	\begin{equation}
	\mathcal{H}^s(K)=\lim_{\delta\rightarrow 0} \inf\Big\{\sum_{\gamma=1}^\infty|x_\gamma|^s\mid \{x_\gamma\} \, \text{be a $\delta$-covering for $K$} \Big\}.\label{def3}
	\end{equation}
\end{definition} As it is well known, the $s$ parameter may be related to the Hausdorff dimension when it reaches a critical value. In the sequel, it shall have no impact on our discussion, except when we recover the standard case, where its interpretation is the usual one. 

Consider now, $\{U_\alpha\}$ a disjoint countable covering of open sets covering $^\varphi\mathbb{R}^{4}$, and $\{U_{\alpha\beta}\}$ a disjoint countable refinement of open sets $\cup_{\beta=1}^\infty U_\beta$ covering a given $U_\alpha$. In the sequel, define what we call equipped local set $(1-\varphi_\alpha)U_\alpha:=\cup_{\beta=1}^\infty(1-\varphi_{\alpha\beta})U_{\alpha\beta}$, where $(1-\varphi_{\alpha\beta})U_{\alpha\beta}=\{(1-\varphi_{\alpha\beta}(x))(x)\mid x\in U_\beta\subset U_\alpha\}$ and $\varphi \in (0,1)$ for $x\in\,^\varphi\mathbb{R}^4$ is the deformation function. Let us elaborate on the Hausdorff measure associated with the equipped local set. By construction, $(1-\varphi_{\alpha\beta})U_{\alpha\beta}\simeq \mathbb{R}^4$ (locally, as it is clear) for each $\beta$ and then, since the covering is assumed disjoint, we have       	 
\begin{eqnarray}
\mathcal{H}^s((1-\varphi_\alpha)U_\alpha)=\sum_{\beta=1}^\infty\mathcal{H}^s((1-\varphi_{\alpha\beta})U_{\alpha\beta}). \label{ape1}
\end{eqnarray} It is straightforward to see that the deformation function contracts the equipped local set. In the limit of $\varphi_{\alpha\beta}\rightarrow 1$ for some $U_\beta$, the equipped local set is shrunk to the empty set and, due to the presence of supremum and infimum in the measure definition, $\mathcal{H}^s((1-\varphi_\alpha)U_\alpha)$ became ill-defined. Some problem is to be expected in the frontier between $^\varphi\mathbb{R}^4$ and $\Sigma_1^E$, where the superscript $E$ stands for the Euclidean version of $\Sigma_1$. In any case, the equipped local set concept of Hausdorff measure, as we built it, does not apply to $\Sigma_1^E$ and is quite usual to $\Sigma_3^E$ (where $\varphi$ is constant different from $1$ for every open set). The novelty and usefulness of this concept to our purposes is regarded to $^\varphi\mathbb{R}^4$. 

To further proceed, take $J$ as the counter set of $\beta$ values for which $\varphi_{\alpha\beta}$ is constant ($0<1-\varphi_{\alpha(\beta\in J)}<1$) in $U_\beta$. Therefore, the $J$ set of index selects those elements of the covering $\{U_\beta\}$ with constant deformation, though different open sets may have different constant deformation values. Now, we shall adapt a standard result to our case.  
\begin{prop}
$\mathcal{H}^s((1-\varphi_\alpha)U_\alpha)=\sum_{\beta\in J}(1-\varphi_{\alpha\beta})^s\mathcal{H}^s(U_{\alpha\beta})+\sum_{\beta\notin J}\mathcal{H}^s((1-\varphi_{\alpha\beta})U_{\alpha\beta})$. 
\end{prop}	
\begin{proof}
	Starting from the very definition (\ref{def3}), we have
	\begin{equation}
	\sum_{\beta\in J}\mathcal{H}^s((1-\varphi_{\alpha\beta})U_{\alpha\beta})=\sum_{\beta\in J}\lim_{\delta\rightarrow 0}\inf\Big\{\sum_{\gamma=1}^\infty|(1-\varphi_{\alpha\beta})x_\gamma|^s \Big\},
	\end{equation} which amounts out to 
	\begin{equation}
	\sum_{\beta\in J}\mathcal{H}^s((1-\varphi_{\alpha\beta})U_{\alpha\beta})=\sum_{\beta\in J}(1-\varphi_{\alpha\beta})^s\lim_{\delta\rightarrow 0}\inf\Big\{\sum_{\gamma=1}^\infty|x_\gamma|^s \Big\}. \label{app2}
	\end{equation} Now, from the $\delta$-covering condition, one sees that $0<|x_\gamma|<\delta/(1-\varphi_{\alpha\beta})$ and, with the conditions in vogue, we have (calling $\xi:=\delta/(1-\varphi_{\alpha\beta})$) that $\lim_{\delta\rightarrow 0}(\cdots)=\lim_{\xi\rightarrow 0}(\cdots)$, as far as the limit exists. Therefore every $\delta$-covering of $(1-\varphi_{\alpha\beta})U_{\alpha\beta}$ can be written as a $\xi$-covering of $U_{\alpha\beta}$, enabling to change the limit in (\ref{app2}) to $\xi$, arriving at 
	\begin{equation}
	\sum_{\beta\in J}\mathcal{H}^s((1-\varphi_{\alpha\beta})U_{\alpha\beta})=\sum_{\beta\in J}(1-\varphi_{\alpha\beta})^s\mathcal{H}^s(U_{\alpha\beta}).
	\end{equation}
\end{proof}	To finalize, we would like to point out a simple criteria for slowly varying deformation, based on the previous considerations: for $x\in U_\beta\subset U_\alpha$ a slowly varying $\varphi_{\alpha\beta}(x)$ function is (in an excellent approximation) constant for every $U_\beta$. Since the deformation may be a different constant for different open sets, the slope is low enough that the above sentence is still valid. In other words, we shall consider as exhaustive criteria of slowly varying deformation every $\beta$ belonging to $J$, so that the equipped local set Hausdorff measure may be entirely written as $\sum_{\beta\in J}(1-\varphi_{\alpha\beta})^s\mathcal{H}^s(U_{\alpha\beta})$. These considerations suffice to support the assumptions made in section IV.

\end{document}